\documentclass[preprint,12pt]{elsarticle}

\usepackage{amssymb}
\usepackage{verbatim}

\usepackage{cancel}
\usepackage{enumitem} 
\setlist{noitemsep,topsep=0pt,parsep=0pt} 

\usepackage[margin=1cm]{caption} 
\usepackage{subfig}
\usepackage{mathtools} 

\usepackage{tikz}
\usetikzlibrary{arrows,backgrounds,calc,fit,decorations.pathreplacing,decorations.markings,shapes.geometric}

\tikzset{every fit/.append style=text badly centered}

\usepackage{tyson}

\usepackage[bookmarks=true,hypertexnames=false]{hyperref}
\hypersetup{colorlinks=true, citecolor=blue, linkcolor=red, urlcolor=blue}

\newcommand{\Holant}{\operatorname{Holant}}

\newcommand{\holant}[2]{\ensuremath{\Holant\left(#1\mid #2\right)}}

{\par\medskip}

\tikzstyle{internal} = [draw, fill, shape=circle]
\tikzstyle{external} = [shape=circle]
\tikzstyle{square}   = [draw, fill, rectangle]
\tikzstyle{triangle} = [draw, fill, regular polygon, regular polygon sides=3, inner sep=3pt]
\tikzstyle{pentagon} = [draw, fill, regular polygon, regular polygon sides=5, inner sep=2pt, minimum size=14pt]

\journal{Information and Computation}

\begin{document}

\begin{frontmatter}



\title{Complexity Classification Of The Six-Vertex Model}


\author[1]{Jin-Yi Cai}
\author[2]{Zhiguo Fu}
\author[3]{Mingji Xia}


\address[1]{Department of Computer Science, University of Wisconsin-Madison.\footnote{{\tt jyc@cs.wisc.edu}}}
\address[2]{School of Mathematics, Jilin University.\footnote{{\tt fuzg@jlu.edu.cn}}}
\address[3]{State Key Laboratory of Computer Science, Institute of Software, Chinese Academy of Science, and University of Chinese Academy of Science, Beijing, China.\footnote{{\tt mingji@ios.ac.cn}. Supported by China National 973 program 2014CB340300}}
\begin{abstract}
We prove a complexity dichotomy theorem for the six-vertex model.
For every setting of the parameters of the model, we prove that
computing  the partition function is either solvable in polynomial time
or \#P-hard. The dichotomy criterion is explicit.
\end{abstract}


\end{frontmatter}


\section{Introduction}\label{sec:intro}

A primary purpose of complexity theory is to provide
classifications to computational problems according to
their inherent computational difficulty. While computational problems
can come from many sources, a class of problems from statistical
mechanics has a remarkable affinity to what is naturally
studied in complexity theory. These are the \emph{sum-of-product}
computations, a.k.a. \emph{partition functions} in physics.

Well-known examples of partition functions from physics that
have been investigated intensively in complexity theory
include the Ising model and Potts model~\cite{jerrum-sinclair,
goldberg-jerrum-patterson,goldberg-jerrum-potts,lu-ying-li}. Most of these are spin systems.
Spin systems as well as the more general counting
constraint satisfaction problems (\#CSP) are special cases
 of Holant problems~\cite{Cai-Lu-Xia}
(see Section~\ref{sec:preliminary} for definitions).
Roughly speaking, Holant problems are tensor networks where
edges of a graph are variables while vertices
are local constraint functions; by contrast, in  spin systems
vertices are variables and edges are (binary) constraint functions.
Spin systems can be simulated easily as Holant problems,
but Freedman, Lov\'{a}sz and Schrijver
proved that simulation in the reverse direction is generally
not possible~\cite{friedman-lovasz-schrjiver}.
In this paper we study a family of partition functions
that fit the Holant problems naturally, but not as a spin system.
This is the \emph{six-vertex model}.

The six-vertex model in statistical mechanics concerns
 crystal lattices with hydrogen bonds.
 Remarkably it can be expressed
perfectly as a family of Holant problems with 6 parameters for the
associated signatures, although in physics people are more
focused on regular structures such as lattice graphs,
and asymptotic limit.
In this paper we study the partition functions of
six-vertex models purely from a complexity theoretic view,
and prove a complete classification
of these Holant problems, where the
6 parameters can be arbitrary complex numbers.

The first model in the family of six-vertex models was introduced
by Linus Pauling in 1935 to account for the residual entropy
of water ice~\cite{Pauling}.
Suppose we have  a large number of oxygen atoms.
Each oxygen atom is connected by a bond to four other neighboring oxygen
atoms, and each bond is occupied by one hydrogen atom between two oxygen
atoms. Physical constraint requires that the hydrogen is
closer to either one or the other of the two neighboring oxygens, but
never in the middle of the bond.
 Pauling argued~\cite{Pauling} that, furthermore, the allowed configuration of
 hydrogen atoms is such that at each oxygen site,  exactly two hydrogens
are closer to it, and the other two are farther away.
The placement of oxygen and hydrogen atoms
 can be naturally represented by vertices and edges of a 4-regular graph.
The constraint on the placement of hydrogens can be represented by
an orientation of the edges of the graph, such that
at every vertex, exactly two edges are oriented toward the vertex,
and exactly two edges are oriented away from it.
In other words, this is an \emph{Eulerian orientation}.
Since there are ${4 \choose 2} =6$ local valid configurations,
this is called the six-vertex model.
In addition to  water ice,
potassium dihydrogen phosphate KH$_2$PO$_4$ (KDP) also
satisfies this model.

The valid local configurations of the six-vertex model
are illustrated in Figure~\ref{fig:6vertex}.
\begin{figure}[hbtp]
	\begin{center}
		\includegraphics[width=\textwidth]{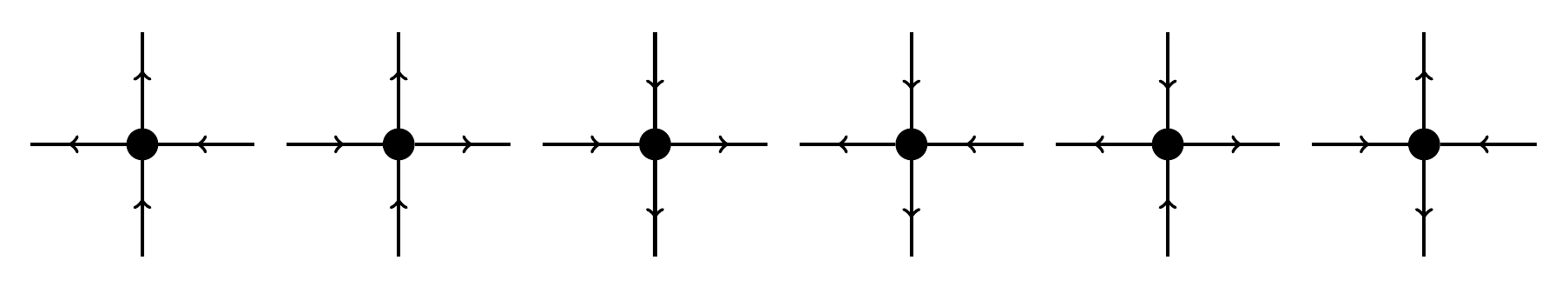}
	\caption{Valid configurations of the six-vertex model.}
	\label{fig:6vertex}
	\end{center}
\end{figure}
There are parameters $\epsilon_1, \epsilon_2, \ldots, \epsilon_6$
associated with each type of the local configuration.
 The total energy $E$ is given by
$E = n_1 \epsilon_1 + n_2 \epsilon_2 + \ldots +  n_6  \epsilon_6$,
where $n_i$ is the number of local configurations of type $i$.
Then the  partition function is $Z = \sum e^{-E/k_BT}$,
where the sum is  over all valid configurations, $k_B$
 is Boltzmann's constant, and $T$ is the system's temperature.
Mathematically, this is a \emph{sum-of-product}
computation where the sum is over all Eulerian orientations
of the graph, and the product is over all vertices where each
vertex contributes a factor $c_i = c^{\epsilon_i}$ if it is in configuration $i$
($1 \le i \le 6$) for some constant $c$.

Some choices of the parameters
are well-studied.
On the square lattice graph,
when modeling ice one takes
$\epsilon_1=\epsilon_2= \ldots =\epsilon_6=0$.
In 1967, Elliott Lieb~\cite{Lieb}
 famously showed that,
as the number $N$ of vertices approaches $\infty$,
the value of the ``partition function per vertex''
$W = Z^{1/N}$ approaches $\left( \frac{4}{3} \right)^{3/2}
\approx 1.5396007\ldots$ (Lieb's square ice constant).
 This matched experimental
data $1.540 \pm 0.001$ so well that it is considered a triumph.

There are other well-known choices in the six-vertex model family.
The KDP model of a ferroelectric is to set
$\epsilon_1=\epsilon_2=0$,  and
$\epsilon_3 = \epsilon_4 = \epsilon_5 = \epsilon_6>0$.
 The Rys $F$ model of an antiferroelectric  is to set
$\epsilon_1=\epsilon_2= \epsilon_3 = \epsilon_4 >0$,
and $\epsilon_5 = \epsilon_6 = 0$.
When there is no ambient electric field, the model chooses
the  zero field assumption:
$\epsilon_1=\epsilon_2$,
$\epsilon_3 = \epsilon_4$, and
$\epsilon_5 = \epsilon_6$.
Historically these are widely considered  among
the most significant applications ever made of statistical mechanics
to real substances.
In classical statistical mechanics the parameters are all
real numbers while in quantum theory the parameters are
complex numbers in general.

In this paper, we give a complete classification
of the complexity of calculating
the partition function $Z$ on any 4-regular graph  defined by an arbitrary
choice parameter values
$c_1, c_2, \ldots, c_6 \in \mathbb{C}$.
(To state our theorem in strict Turing machine model,
we take $c_1, c_2, \ldots, c_6$
to be  \emph{algebraic} numbers.)
Depending on the setting of these values,
we show that the partition function  $Z$ is either
computable in polynomial time, or it is \#P-hard, with nothing
in between.
The dependence of this dichotomy on the
values
$c_1, c_2, \ldots, c_6$ is explicit.

A number of complexity dichotomy theorems for counting problems have
been proved previously. These are mostly on spin systems,
or \#CSP (counting Constraint
Satisfaction Problems),
or on
Holant problems with \emph{symmetric} local constraint functions.
 \#CSP is the special case of Holant problems
where  {\sc Equalities} of all arities are auxiliary functions
assumed to be present.
Spin systems are a further specialization of  \#CSP,
where there is a single binary constraint function
(see Section~\ref{sec:preliminary}).
 The six-vertex model cannot be expressed as
a \#CSP problem. It is a Holant problem where the constraint functions are
\emph{not symmetric}.  Thus previous dichotomy theorems do not apply.
This is the  first complexity dichotomy theorem proved for
a class of Holant problems on \emph{non-symmetric} constraint
functions and without auxiliary functions assumed to be present.

However, one important technical ingredient of our proof is to discover
a direct connection between some subset of the six-vertex models
with spin systems.
Another technical highlight is a new interpolation technique
that carves out subsums of a partition function
 by assembling a suitable
sublattice, and partition the sum over  an exponential range
 according to an enumeration
of the intersections of  cosets of the sublattice with this range.


\section{Preliminaries and notations}\label{sec:preliminary}

A constraint function $f$ of arity $k$
is a map $\{0,1\}^k  \rightarrow \mathbb{C}$.
Fix a set of constraint functions $\mathcal{F}$. A signature grid
$\Omega=(G, \pi)$
 is a tuple, where $G = (V,E)$
is a graph, $\pi$ labels each $v\in V$ with a function
$f_v\in\mathcal{F}$ of arity ${\operatorname{deg}(v)}$,
and the incident edges
$E(v)$ at $v$ with input variables of $f_v$.
 We consider all 0-1 edge assignments $\sigma$,
each gives an evaluation
$\prod_{v\in V}f_v(\sigma|_{E(v)})$, where $\sigma|_{E(v)}$
denotes the restriction of $\sigma$ to $E(v)$. The counting problem on the instance $\Omega$ is to compute
Holant$_{\Omega}=\sum_{\sigma:E\rightarrow\{0, 1\}}\prod_{v\in V}f_v(\sigma|_{E(v)})$.
The Holant problem parameterized by the set $\mathcal{F}$ is denoted by Holant$(\mathcal{F})$. We denote by Holant$(\mathcal{F} \mid \mathcal{G})$
the Holant problem on bipartite graphs where signatures from $\mathcal{F}$
and $\mathcal{G}$ are assigned to vertices from the Left and Right.

A spin system on $G = (V, E)$ has a variable for every $v \in V$
and a binary function $g$ for every edge $e \in E$.
The partition function is $\sum_{\sigma: V \rightarrow\{0, 1\}}
\prod_{(u, v) \in E} g(\sigma(u), \sigma(v))$.
Spin systems are special cases of \#CSP$(\mathcal{F})$
(counting CSP)
where $\mathcal{F}$ consists of a single binary function.
In turn,  \#CSP$(\mathcal{F})$ is the special case of
Holant  where  $\mathcal{F}$ contains {\sc Equality} of all arities.


A constraint function is also called a signature.
A function $f$ of arity $k$ can be represented by listing its values in lexicographical order as in a truth table, which
is a vector in
$\mathbb{C}^{2^{k}}$, or as a tensor in $(\mathbb{C}^2)^{\otimes k}$, or as a matrix in $\mathbb{C}^{2^{k_1}}\times \mathbb{C}^{2^{k_2}}$
if we partition the $k$ variables to two parts,  where $k_1+k_2=k$.
A function is symmetric if its value depends only
on the Hamming weight of its input.
A symmetric function $f$ on $k$ Boolean variables can
be expressed as
$[f_0, f_1, \ldots, f_k]$,
where $f_w$ is the value of $f$ on inputs of Hamming weight $w$.
For example, $(=_k)$ is the {\sc Equality} signature $[1, 0, \ldots, 0, 1]$
(with $k-1$ 0's) of arity $k$.
We use $\neq_2$ to denote binary {\sc Disequality} function $[0,1,0]$.
The support of a function $f$ is the set of inputs on which $f$ is nonzero.

Given an instance $\Omega=(G, \pi)$ of Holant$(\mathcal{F})$, we add a middle point on each edge as a new vertex to $G$, then each edge becomes a path of length two through the new vertex. Extend $\pi$ to label a function $g$ to each new vertex. This gives a bipartite Holant problem Holant$(g\mid \mathcal{F})$. It is obvious  that Holant$(=_2\mid \mathcal{F})$ is equal to Holant$(\mathcal{F})$.

For $T \in {\rm GL}_2({\mathbb{C}})$
 and a signature $f$ of arity $n$, written as
a column vector  $f \in \mathbb{C}^{2^n}$, we denote by
$T^{-1}f = (T^{-1})^{\otimes n} f$ the transformed signature.
  For a signature set $\mathcal{F}$,
define $T^{-1} \mathcal{F} = \{T^{-1}f \mid  f \in \mathcal{F}\}$.
For signatures written as
 row vectors we define $\mathcal{F} T$ similarly.
%
%
The holographic transformation defined by $T$ is the following operation:
given a signature grid $\Omega = (H, \pi)$ of $\holant{\mathcal{F}}{\mathcal{G}}$,
for the same bipartite graph $H$,
we get a new signature grid $\Omega' = (H, \pi')$ of $\holant{\mathcal{F} T}{T^{-1} \mathcal{G}}$ by replacing each signature in
$\mathcal{F}$ or $\mathcal{G}$ with the corresponding signature in $\mathcal{F} T$ or $T^{-1} \mathcal{G}$.

In this paper we focus on Holant$(\neq_2 \mid f)$ when $f$ has support among strings of hamming weight $2$. They are
the six-vertex models on general graphs. This
 corresponds to a set of (non-bipartite)
 Holant problems by a holographic reduction \cite{Val08}. Let $Z=\frac{1}{\sqrt{2}} \left[\begin{smallmatrix}
1 & 1 \\
\mathfrak{i} & -\mathfrak{i}
\end{smallmatrix}\right]$. The matrix form of $\neq_2$ is $\left[\begin{smallmatrix}
0 & 1 \\
 1 & 0
\end{smallmatrix}\right]=Z^{\tt T}Z$.
Under a holographic transformation with bases $Z$, Holant$(\neq_2 \mid f)$ becomes Holant$(=_2 \mid Z^{\otimes 4}f)$, where $Z^{\otimes 4}f$ is a column vector $f$ multiplied by the matrix tensor power $Z^{\otimes 4}$.
The bipartite Holant problems of the form Holant$(\neq_2 \mid f)$
naturally correspond to the non-bipartite Holant problems Holant$(Z^{\otimes 4}f)$.
In general $f$ and $Z^{\otimes 4}f$ are non-symmetric functions.

A signature $f$ of arity 4 has the signature matrix
$M = M_{x_1x_2, x_4x_3}(f)=\left[\begin{smallmatrix}
f_{0000} & f_{0010} & f_{0001} & f_{0011}\\
f_{0100} & f_{0110} & f_{0101} & f_{0111}\\
f_{1000} & f_{1010} & f_{1001} & f_{1011}\\
f_{1100} & f_{1110} & f_{1101} & f_{1111}
\end{smallmatrix}\right]
$.
If $\{i,j,k,\ell\}$ is a permutation of $\{1,2,3,4\}$,
then the $4 \times 4$ matrix $M_{x_ix_j, x_kx_{\ell}}(f)$ lists the 16 values
 with row  index $x_ix_j \in\{0, 1\}^2$
and column index
$x_kx_{\ell} \in\{0, 1\}^2$ in lexicographic order.

Let $\mathbb{N} = \{0,1,2,\ldots\}$,
$H=\frac{1}{\sqrt{2}}
\left[\begin{smallmatrix}
1 & 1 \\
 1 & -1
\end{smallmatrix}\right]
$ and
$N=
\left[\begin{smallmatrix}
0 & 1 \\
 1 & 0
\end{smallmatrix}\right]
\otimes
\left[\begin{smallmatrix}
0 & 1 \\
 1 & 0
\end{smallmatrix}\right]
=
\left[\begin{smallmatrix}
 0 & 0 & 0 & 1 \\
 0 & 0 & 1 & 0 \\
 0 & 1 & 0 & 0 \\
 1 & 0 & 0 & 0 \\
\end{smallmatrix}\right]$.
Note that $N$ is the double {\sc Disequality},
which is the function of connecting two pairs of edges by $(\not =_2)$.

If $f$ and $g$ have signature matrices
 $M(f) = M_{x_ix_j, x_k x_{\ell}}(f)$ and $M(g) =
M_{x_{s}x_{t}, x_{u}x_{v}}(g)$,
by connecting $x_k$ to $x_{s}$, $x_{\ell}$ to $x_t$,
both with  {\sc Disequality}  $(\not =_2)$, we get a signature
of arity 4 with the signature matrix
$M(f) N M(g)$
by matrix product
with row index $x_ix_j$ and column index $x_{u}x_{v}$.

The six-vertex model is Holant$( \not =_2 \mid f)$,
where $M_{x_1x_2, x_4x_3}(f) =
\left[\begin{smallmatrix}
0 & 0 & 0 & a \\
0 & b & c & 0 \\
0 & z & y & 0 \\
x & 0 & 0 & 0
\end{smallmatrix}\right]$.
We also write this matrix by $M(a,x,b,y,c,z)$. When $a=x, b=y$ and $c=z$, we abridge it as $M(a,b,c)$.
Note that all nonzero entries of $f$ are on Hamming weight 2.
Denote the 3 pairs of
 ordered complementary strings
 by $\lambda = 0011, \overline{\lambda} = 1100$,
$\mu = 0110, \overline{\mu} = 1001$, and $\nu = 0101, \overline{\nu}
= 1010$.
The support of $f$ is the union
$\{\lambda, \overline{\lambda}, \mu, \overline{\mu}, \nu, \overline{\nu}\}$
 of the pairs $(\lambda, \overline{\lambda})$,
$(\mu, \overline{\mu})$ and $(\nu, \overline{\nu})$,
on which $f$ has values
$(a,x), (b,y)$ and $(c,z)$.
If $f$ has the same value in  a pair,
 say $a=x$ on $\lambda$ and $\overline{\lambda}$,
 we say it is a twin.

The permutation group $S_4$ on $\{x_1, x_2, x_3, x_4\}$
induces a group action on $\{s \in \{0, 1\}^4
\mid {\rm wt}(s) =2\}$ of size 6.
This is a faithful representation of $S_4$ in $S_6$.
Since the action of $S_4$ preserves complementary pairs,
this group action has nontrivial blocks of imprimitivity,
namely $\{A,B,C\} =
\{\{\lambda, \overline{\lambda}\},
\{\mu, \overline{\mu}\}, \{\nu, \overline{\nu}\}\}$.
The action on the blocks is a homomorphism of $S_4$ onto
$S_3$, with kernel $K = \{1, (12)(34), (13)(24), (14)(23)\}$.
In particular one can calculate that the subgroup
$S_{\{2,3,4\}} = \{1, (23), (34), (24), (243), (234)\}$
maps to $\{1, (AC), (BC), (AB), (ABC), (ACB)\}$.
By a permutation from $S_4$,
we may permute the matrix $M(a,x,b,y,c,z)$
by any permutation on the values
$\{a,b,c\}$ with the corresponding permutation on $\{x,y,z\}$,
and moreover we can further flip
an even number
of pairs $(a,x)$, $(b,y)$
and $(c,z)$.
In particular, we can
arbitrarily reorder the  three rows in
$\left[\begin{smallmatrix}
a & x \\
b & y \\
c & z
\end{smallmatrix}\right]$, and
we can also reverse the order of
 arbitrary two rows together. In the proof, after one construction, we may use this property to get a
 similar construction and conclusion, by quoting this
symmetry of three pairs or six values.

 \begin{definition}
 A 4-ary signature is redundant iff in its  4 by 4 signature matrix
  the middle two rows are identical and the
 middle two columns are identical.
 \end{definition}
\begin{theorem}\cite{caiguowilliams13}\label{redundant}
If $f$ is a redundant signature and the determinant
\[ \det \left[\begin{matrix}
f_{0000} & f_{0010} & f_{0011}\\
f_{0100} & f_{0110} & f_{0111}\\
f_{1100} & f_{1110} & f_{1111}
\end{matrix}\right] \not = 0, \]
 then $\operatorname{Holant}(\neq_2|f)$ is \#P-hard.
\end{theorem}

We use $\mathscr{A}$ and $\mathscr{P}$ to denote two classes of tractable signatures. The classes
$\mathscr{A}$ and $\mathscr{P}$ are
identified as tractable for \#CSP \cite{cailuxia-2014}. Problems defined by $\mathscr{A}$ are tractable essentially by Gauss Sums (See Theorem~6.30 of
\cite{caichenliptonlu}).
The signatures in $\mathscr{P}$ are tensor products of signatures whose supports are among two complementary
bit vectors. Problems defined by them are tractable by a propagation algorithm.
The full
version \cite{caifuguowilliam} contains complete definitions and characterizations of these classes.

\begin{theorem}\cite{cailuxia-2014}\label{csp:dichotomy}
Let $\mathcal{F}$ be any set of complex-valued signatures in Boolean variables. Then $\operatorname{\#CSP}(\mathcal{F})$
is \#P-hard unless
$\mathcal{F}\subseteq\mathscr{A}$ or
$\mathcal{F}\subseteq\mathscr{P}$,
  in which case the problem is computable in polynomial time.
\end{theorem}

\begin{definition}
$\mathscr{M}$ is the set of functions, whose support is composed of strings of Hamming weight at most one.
$\mathscr{M'}=\{g \mid \exists f \in \mathscr{M}, ~ g(\mathbf{x})=f(\overline{\mathbf{x}})\}$, where $\overline{\mathbf{x}}$ is the complement of $\mathbf{x}$.
\end{definition}

Note that all unary functions are in $\mathscr{M} \cap \mathscr{M'}$.
Theorem \ref{holant*:dichotomy} is a consequence of Theorem 2.2 in~\cite{cailuxia-2011}.

\begin{theorem}\label{holant*:dichotomy}
$\operatorname{Holant}(\neq_2\mid\mathscr{M} )$
and  $\operatorname{Holant}(\neq_2\mid \mathscr{M'} )$
are polynomial
time computable.
\end{theorem}

\section{Main theorem}

\begin{theorem}
Let $f$ be a 4-ary signature with the signature matrix
$M_{x_1x_2, x_4x_3}(a, x, b, y, c, z)$, then
Holant$(\neq_2\mid f)$ is \#P-hard except for the following cases:
\begin{itemize}
\item $f\in\mathscr{P}$;
\item $f\in\mathscr{A}$;
\item there is a zero in each pair $(a,x), (b,y), (c,z)$;
\end{itemize}
in which cases Holant$(\neq_2\mid f)$ is computable in polynomial time.
\end{theorem}

We prove the complexity classification by categorizing the
six values $a,b,c,x,y,z$  in the following way.

\begin{enumerate}
\item\label{one-zero-pair}
 There is a zero pair. If $f\in\mathscr{A}\cup \mathscr{P}$,
then  it is tractable. Otherwise it is  \#P-hard.


\item\label{all-nonzero-case}
All values in  $\{a,x,b,y,c,z\}$ are nonzero. We prove these are \#P-hard.
\begin{enumerate}
\item\label{subcase-a}  Three twins. We prove this case mainly by an
interpolation reduction from redundant signatures, then
apply Theorem~\ref{redundant}.
\item  There is one pair that is not twin. We prove this by a
reduction from Case~\ref{subcase-a}.
\end{enumerate}


\item\label{exactly-one-zero}
There is exactly one zero in $\{a,x,b,y,c,z\}$. All are \#P-hard by reducing from Case~\ref{all-nonzero-case}.


\item\label{exactly-two-zeros}
There are exactly two zeros which are from different pairs. All are \#P-hard by reducing from Case~\ref{all-nonzero-case}.


\item\label{one-zero-in-each-pair}
There is one zero in each pair. These are tractable according to Theorem \ref{holant*:dichotomy}.
\end{enumerate}

By definition, in
 Case~\ref{one-zero-pair} and Case~\ref{one-zero-in-each-pair}, 
$f$ may have more zero values than the stated ones. 

These cases above cover all possibilities:
After Case~\ref{one-zero-pair}  we may assume that there is no
zero pair. Then after Case~\ref{all-nonzero-case} 
we may assume there is at least one zero and there is no
zero pair. Similarly after Case~\ref{exactly-one-zero} we may assume 
there are at least two zeros and there is no
zero pair. So Case~\ref{exactly-two-zeros} finishes off the case 
when there are exactly two zeros.
 After Case~\ref{exactly-two-zeros} we may assume  there are at least
three zeros, but there is no zero pair. Therefore
we may assume the only case remaining is where there are
exactly three zeros in three distinct
pairs, and Case~\ref{one-zero-in-each-pair} finishes the proof.



In the following  we prove the 5 cases to prove the main theorem.

\section{Case~\ref{one-zero-pair}: One zero pair}
In this section we prove Case~\ref{one-zero-pair}.
Note that by renaming the variables $x_1, x_2, x_3, x_4$
we may assume the signature $f$ of arity 4 with one zero pair has the  form in
(\ref{one-zero-pair-matrix-form}).

\begin{lemma}\label{spin}
Let $f$ be a 4-ary signature with the signature matrix
\begin{equation}\label{one-zero-pair-matrix-form}
M_{x_sx_t, x_ux_v}(f)=\begin{bmatrix}
0 & 0& 0& 0\\
0 & \alpha & \beta & 0\\
0 & \gamma & \delta & 0\\
0 & 0& 0& 0\\
\end{bmatrix},
\end{equation}
where $\{s,t,u,v\}$ is a permutation of $\{1, 2, 3, 4\}$.
 Then $\operatorname{Holant}(\neq_2\mid f)$
is \#P-hard unless
$f\in\mathscr{A}$ or
$f\in\mathscr{P}$,
  in which case the problem is computable in polynomial time.
\end{lemma}
\begin{proof}
By the $S_4$ group symmetry, we only
need to prove the lemma for $(s, t, u, v)=(1, 2, 4, 3)$.
Tractability follows from Theorem~\ref{csp:dichotomy}.

Let $g(x, y)$ be the binary signature
  $g = \left[\begin{smallmatrix}
 \alpha & \beta \\
 \gamma & \delta
\end{smallmatrix}\right]$
in matrix form.
This means
$g_{00} =  \alpha, g_{01} =  \beta, g_{10} = \gamma$ and $g_{11} = \delta$.
We prove that
$\operatorname{\#CSP}(g)\leq_{\operatorname{T}}\operatorname{Holant}(\neq_2\mid f)$
in two steps.
In each step, we begin with a signature grid and end with a new signature grid such that the Holant values of both signature grids are the same.

For step one, let $G=(U, V, E)$ be a bipartite graph
representing  an instance of \#CSP$(g)$, where each $u \in U$ is a variable,
and each $v \in V$ has degree two and is labeled $g$.
  We define a cyclic order of the edges incident to each vertex
$u \in U$, and decompose $u$ into $k = \deg(u)$ vertices.
 Then we connect the $k$
edges originally incident to $u$ to these $k$ new vertices so that each vertex is incident to exactly one edge. We also connect these $k$ new
 vertices in a cycle according to the cyclic order
(see Figure~\ref{holant-csp-b}).
Thus, in effect we have replaced $u$ by a cycle of length $k = \deg(u)$.
(If $k=1$ there is a self-loop.)
    Each of  $k$ vertices has degree 3, and we assign them $(=_3)$.
Clearly this does not change the value of the partition function. The resulting graph has the following properties: (1) every vertex has either degree 2 or degree 3; (2) each degree 2 vertex is connected
to degree 3 vertices;
(3) each degree 3 vertex is connected to exactly one degree 2 vertex.

Now step two. For every $v\in V$, $v$ has degree 2 and
is labeled by $g$. We contract the two edges incident to $v$. The resulting graph $G'=(V', E')$ is 4-regular.
We put a node on every edge of $G'$ and assign $(\neq_2)$ to the node
(see Figure~\ref{holant-csp-c}).
Next we assign a copy of $f$ to every $v'\in V'$ after this contraction.
The input variables $x_1, x_2, x_3, x_4$ are carefully
  assigned at each  copy of $f$ as illustrated in Figure~\ref{four-cases} so that there are exactly two configurations
 to each original cycle, which correspond to cyclic orientations,
  due to the $(\neq_2)$ on it and the
support set of $f$. These correspond to  the 0-1 assignment values at
the original variable $u \in U$.
Moreover in each case, the value of the function $g$ is
perfectly mirrored by the value of the function $f$
under the orientations.
So we have
$\operatorname{\#CSP}(g)\leq _{\operatorname{T}}
\operatorname{Holant}(\neq_2\mid f).$

%

We have $f(x_1,x_2,x_3,x_4)=g(x_1,x_4)\cdot \chi_{x_1 \neq x_2} \cdot \chi_{x_3 \neq x_4}$. Hence, $g \in \mathscr{A} \cup \mathscr{P}$ implies $f \in \mathscr{A} \cup \mathscr{P}$. Therefore if $f \not \in \mathscr{A} \cup \mathscr{P}$,
then $g \not \in \mathscr{A} \cup \mathscr{P}$. Then   $\operatorname{\#CSP}(g)$ is \#P-hard by Theorem~\ref{csp:dichotomy}.
It follows that $\operatorname{Pl-Holant}(\neq_2\mid  f)$ is \#P-hard.
This finishes the proof.
\end{proof}

\section{Case~\ref{all-nonzero-case}: All six values are nonzero}
In this section, we handle the case $axbycz \neq 0$, by proving all problems in this case are \#P-hard.  Firstly, we give a technical lemma for interpolation reduction.
Then we prove the 3-twins case. Finally, we prove the other cases by realizing a 3-twins problem.

\begin{lemma} \label{lem:interpolation}
Suppose $\alpha, \beta \in \mathbb{C} - \{0\}$, and
the lattice  $L = \{(j,k) \in \mathbb{Z}^{2} \mid  \alpha^j \beta^k=1\}$
has the form $L=\{(ns,nt) \mid n \in \mathbb{Z}\}$, where $s,t \in \mathbb{Z}$
 and $(s,t) \neq (0,0)$. Let  $\phi$ and $\psi$ be any numbers
satisfying $\phi^s \psi^t=1$.
If we are  given the values $N_\ell =
\sum_{j,k \ge 0,~j+k\leq m} (\alpha^j \beta^k)^{\ell} x_{j,k}$
for ${\ell}=1,2,  \ldots {m+2 \choose 2}$,
then we can compute $\sum_{j,k \ge 0,~j+k\leq m} \phi^j \psi^k x_{j,k}$
in polynomial time.
\end{lemma}

\clearpage\thispagestyle{empty}\addtocounter{page}{-1}
\begin{figure}
\centering
\subfloat[]{
\begin{tikzpicture}[scale=0.4]
\node [internal, scale=0.6] at (1, 1.73) {};
\node at (1.8, 1.73) {$u'$};
\node [internal, scale=0.6] at (1, 5.73) {};
\node at (1.8, 5.73) {$u$};
\node [external] at (2, 0) {};
\node [external] at (-0.4, 4.33) {};
\node [external] at (2.4, 4.33) {};
\node [external] at (-0.4, 7.13) {};
\node [external] at (2.4, 7.13) {};
\node [external] at (0, 0) {};
\draw (0, 0)--(1, 1.73) [postaction={decorate, decoration={
                                        markings,
                                        mark=at position 0.63 with {\arrow[>=diamond,white] {>}; },
                                        mark=at position 0.63 with {\arrow[>=open diamond]  {>}; } } }];
\draw (2, 0)--(1, 1.73) [postaction={decorate, decoration={
                                        markings,
                                        mark=at position 0.63 with {\arrow[>=diamond,white] {>}; },
                                        mark=at position 0.63 with {\arrow[>=open diamond]  {>}; } } }];
\draw (1, 5.73)--(-0.4, 4.33) [postaction={decorate, decoration={
                                        markings,
                                        mark=at position 0.63 with {\arrow[>=diamond,white] {>}; },
                                        mark=at position 0.63 with {\arrow[>=open diamond]  {>}; } } }];
\draw (1, 5.73)--(-0.4, 7.13) [postaction={decorate, decoration={
                                        markings,
                                        mark=at position 0.63 with {\arrow[>=diamond,white] {>}; },
                                        mark=at position 0.63 with {\arrow[>=open diamond]  {>}; } } }];
\draw (1, 5.73)--(2.4, 4.33) [postaction={decorate, decoration={
                                        markings,
                                        mark=at position 0.63 with {\arrow[>=diamond,white] {>}; },
                                        mark=at position 0.63 with {\arrow[>=open diamond]  {>}; } } }];
\draw (1, 5.73)--(2.4, 7.13) [postaction={decorate, decoration={
                                        markings,
                                        mark=at position 0.63 with {\arrow[>=diamond,white] {>}; },
                                        mark=at position 0.63 with {\arrow[>=open diamond]  {>}; } } }];
\draw (1, 5.73)--(1, 1.73) [postaction={decorate, decoration={
                                        markings,
                                        mark=at position 0.63 with {\arrow[>=diamond,white] {>}; },
                                        mark=at position 0.63 with {\arrow[>=open diamond]  {>}; } } }];
\end{tikzpicture}
}
\qquad
\subfloat[]{
\begin{tikzpicture}[scale=0.4]
\draw [thick, ->](2, 3.73)--(4, 3.73);
\draw [thick, ->](15.5, 3.73)--(17.5, 3.73);
\node [internal, scale=0.6] (51) at (8.2, 7.56) {};
\node [internal, scale=0.6] (52) at (10.8, 7.56) {};
\node [internal, scale=0.6]  (53) at (9.5, 4.23) {};
\node [internal, scale=0.6] (54) at (7.5, 5.56) {};
\node [internal, scale=0.6] (55) at (11.5, 5.56) {};
\node at (9.5, 6) {$u$};
\node at (9.4, 1.2) {$u'$};
\node [external] (51a) at (6.8, 8.96) {};
\node [external] (52a) at (12.4, 8.96) {};
\node [external] (54a) at (6.1, 4.16) {};
\node [external] (55a) at (12.9, 4.16) {};
\draw (51)--(52);
\draw (51)--(54);
\draw (55)--(52);
\draw (53)--(54);
\draw (53)--(55);
\draw (51)--(51a) [postaction={decorate, decoration={
                                        markings,
                                        mark=at position 0.83 with {\arrow[>=diamond,white] {>}; },
                                        mark=at position 0.83 with {\arrow[>=open diamond]  {>}; } } }];
\draw (52)--(52a) [postaction={decorate, decoration={
                                        markings,
                                        mark=at position 0.83 with {\arrow[>=diamond,white] {>}; },
                                        mark=at position 0.83 with {\arrow[>=open diamond]  {>}; } } }];
\draw (54)--(54a) [postaction={decorate, decoration={
                                        markings,
                                        mark=at position 0.83 with {\arrow[>=diamond,white] {>}; },
                                        mark=at position 0.83 with {\arrow[>=open diamond]  {>}; } } }];
\draw (55)--(55a) [postaction={decorate, decoration={
                                        markings,
                                        mark=at position 0.83 with {\arrow[>=diamond,white] {>}; },
                                        mark=at position 0.83 with {\arrow[>=open diamond]  {>}; } } }];
\node [internal, scale=0.6]  (31) at (9.5, 2.23) {};
\node [internal, scale=0.6]  (32) at (8.5, 0.53) {};
\node [internal, scale=0.6]  (33) at (10.5, 0.53) {};
\node [external]  (32a) at (7.1, -0.87) {};
\node [external]  (33a) at (11.9, -0.87) {};
\draw (31)--(32);
\draw (31)--(33);
\draw (32)--(33);
\draw (32)--(32a) [postaction={decorate, decoration={
                                        markings,
                                        mark=at position 0.83 with {\arrow[>=diamond,white] {>}; },
                                        mark=at position 0.83 with {\arrow[>=open diamond]  {>}; } } }];
\draw (33)--(33a) [postaction={decorate, decoration={
                                        markings,
                                        mark=at position 0.83 with {\arrow[>=diamond,white] {>}; },
                                        mark=at position 0.83 with {\arrow[>=open diamond]  {>}; } } }];
\draw (53)--(31) [postaction={decorate, decoration={
                                        markings,
                                        mark=at position 0.83 with {\arrow[>=diamond,white] {>}; },
                                        mark=at position 0.83 with {\arrow[>=open diamond]  {>}; } } }];
%
\end{tikzpicture}
\label{holant-csp-b}
}
\qquad
\subfloat[]{
\begin{tikzpicture}[scale=0.4]
\node [triangle, scale=0.4] (c51) at (18.2, 7.56) {};
\node [external] (c51a) at (18.4, 9.56) {};
\node [external] (c51b) at (17, 9.2) {};
\draw [densely dashed] (c51) to [bend left=20] (c51a);
\draw [densely dashed] (c51) to [bend right=20] (c51b);
\node [triangle, scale=0.4] (c52) at (20.8, 7.56) {};
\node [external] (c52a) at (20.6, 9.56) {};
\node [external] (c52b) at (22, 9.2) {};
\draw [densely dashed] (c52) to [bend right=20] (c52a);
\draw [densely dashed] (c52) to [bend left=20] (c52b);
\node [triangle, scale=0.4]  (c53) at (19.5, 4.23) {};
\node [triangle, scale=0.4] (c54) at (17.5, 5.56) {};
\node [external] (c54a) at (15.5, 5.56) {};
\node [external] (c54b) at (16, 4.16) {};
\draw [densely dashed] (c54) to [bend left=20] (c54a);
\draw [densely dashed] (c54) to [bend right=20] (c54b);
\node [triangle, scale=0.4] (c55) at (21.5, 5.56) {};
\node [external] (c55a) at (23.5, 5.56) {};
\node [external] (c55b) at (23, 4.16) {};
\draw [densely dashed] (c55) to [bend right=20] (c55a);
\draw [densely dashed] (c55) to [bend left=20] (c55b);
\node [external] (c51a) at (16.8, 8.96) {};
\node [external] (c52a) at (22.4, 8.96) {};
\node [external] (c54a) at (16.1, 4.16) {};
\node [external] (c55a) at (22.9, 4.16) {};
\draw (c51) to [bend right=40](c52)[postaction={decorate, decoration={
                                        markings,
                                        mark=at position 0.5 with {\arrow[>=square,white, scale=0.7] {>}; },
                                        mark=at position 0.5 with {\arrow[>=open square, scale=0.7]  {>}; } } }];
\draw (c51) to [bend left=40](c54) [postaction={decorate, decoration={
                                        markings,
                                        mark=at position 0.5 with {\arrow[>=square,white, scale=0.7] {>}; },
                                        mark=at position 0.5 with {\arrow[>=open square, scale=0.7]  {>}; } } }];
\draw (c55) to [bend left=40](c52) [postaction={decorate, decoration={
                                        markings,
                                        mark=at position 0.5 with {\arrow[>=square,white, scale=0.7] {>}; },
                                        mark=at position 0.5 with {\arrow[>=open square, scale=0.7]  {>}; } } }];
\draw (c53) to [bend right=40](c54)[postaction={decorate, decoration={
                                        markings,
                                        mark=at position 0.5 with {\arrow[>=square,white, scale=0.7] {>}; },
                                        mark=at position 0.5 with {\arrow[>=open square, scale=0.7]  {>}; } } }];
\draw (c53) to [bend left=40](c55)[postaction={decorate, decoration={
                                        markings,
                                        mark=at position 0.5 with {\arrow[>=square,white, scale=0.7] {>}; },
                                        mark=at position 0.5 with {\arrow[>=open square, scale=0.7]  {>}; } } }];
\node [triangle, scale=0.4]  (c32) at (18, 2.53) {};
\node [external]  (c32a) at (16, 2.53) {};
\node [external]  (c32b) at (16.5, 1.13) {};
\draw [densely dashed] (c32) to [bend left=20](c32a);
\draw [densely dashed] (c32) to [bend right=20](c32b);
\node [triangle, scale=0.4]  (c33) at (21, 2.53) {};
\node [external]  (c33a) at (23, 2.53) {};
\node [external]  (c33b) at (22.5, 1.13) {};
\draw [densely dashed] (c33) to [bend right=20](c33a);
\draw [densely dashed] (c33) to [bend left=20](c33b);
\draw (c53) to [bend left=20](c32) [postaction={decorate, decoration={
                                        markings,
                                        mark=at position 0.5 with {\arrow[>=square,white, scale=0.7] {>}; },
                                        mark=at position 0.5 with {\arrow[>=open square, scale=0.7]  {>}; } } }];
\draw (c32) to [bend left=20](c33) [postaction={decorate, decoration={
                                        markings,
                                        mark=at position 0.5 with {\arrow[>=square,white, scale=0.7] {>}; },
                                        mark=at position 0.5 with {\arrow[>=open square, scale=0.7]  {>}; } } }];
                                        \draw (c53) to [bend right=20](c33) [postaction={decorate, decoration={
                                        markings,
                                        mark=at position 0.5 with {\arrow[>=square,white, scale=0.7] {>}; },
                                        mark=at position 0.5 with {\arrow[>=open square, scale=0.7]  {>}; } } }];
\end{tikzpicture}
\label{holant-csp-c}
}
\caption{\scriptsize{The reduction from \#CSP$(g)$ to Holant$(\neq_2|f)$.
 The circle vertices are assigned $=_d$, where $d$ is the degree of the corresponding vertex,
 the diamond vertices are assigned $g$,
 the triangle vertices are assigned $f$,
 and the square vertices are assigned $\neq_2$.
 In the first step, we replace a vertex by a cycle, where the length of the cycle is the degree of the vertex.
 The vertices on the cycle are assigned $=_3$.
 In the second step, we merge two vertices that are connected to the diamond with $g$
and assign $f$ to the new vertex.}}
 \label{holant-csp}
\end{figure}
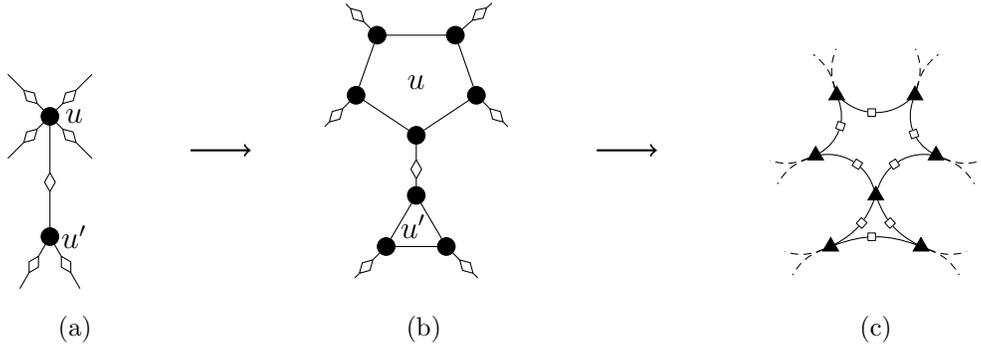

\vspace{-5in}
\input{4cases}
\vspace{-5in}
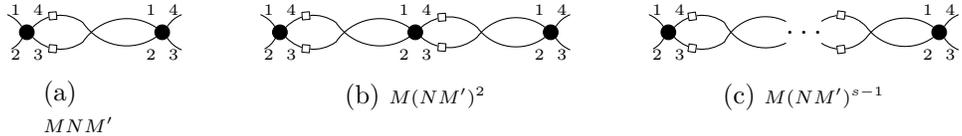
\begin{figure}
\centering
\subfloat[${\scriptstyle MNM'}$]{
\begin{tikzpicture}[scale=0.3]
\node [internal, scale=0.5] (0) at (0, 2) {};
\node at (-0.5, 3) {\tiny{1}};
\node at (-0.5, 1) {\tiny{2}};
\node at (0.5, 3) {\tiny{4}};
\node at (0.5, 1) {\tiny{3}};
\node [external] (1) at (3, 2) {};
\node [external] (0a) at (-1.5, 3) {};
\node [external] (0b) at (-1.3, 1) {};
\node [internal, scale=0.5] (2) at (6, 2) {};
\node at (5.5, 3) {\tiny{1}};
\node at (5.5, 1) {\tiny{2}};
\node at (6.5, 3) {\tiny{4}};
\node at (6.5, 1) {\tiny{3}};
\node [external] (2a) at (7.5, 3) {};
\node [external] (2b) at (7.5, 1) {};
\draw (0) to [bend left=40] (1) to [bend right=45] (2)[postaction={decorate, decoration={
                                        markings,
                                        mark=at position 0.2 with {\arrow[>=square,white, scale=0.7] {>}; },
                                        mark=at position 0.2 with {\arrow[>=open square, scale=0.7]  {>}; } } }];
\draw (0) to [bend right=40] (1) to [bend left=45] (2)[postaction={decorate, decoration={
                                        markings,
                                        mark=at position 0.2 with {\arrow[>=square,white, scale=0.7] {>}; },
                                        mark=at position 0.2 with {\arrow[>=open square, scale=0.7]  {>}; } } }];
\draw (0) to [bend right=15] (0a);
\draw (0) to [bend left=15] (0b);
\draw (2) to [bend right=15] (2b);
\draw (2) to [bend left=15] (2a);
\end{tikzpicture}
}
\subfloat[${\scriptstyle M(NM')^2}$]{
\begin{tikzpicture}[scale=0.3]
\node [internal, scale=0.5] (0) at (0, 2) {};
\node at (-0.5, 3) {\tiny{1}};
\node at (-0.5, 1) {\tiny{2}};
\node at (0.5, 3) {\tiny{4}};
\node at (0.5, 1) {\tiny{3}};
\node [external] (1) at (3, 2) {};
\node [external] (0a) at (-1.5, 3) {};
\node [external] (0b) at (-1.3, 1) {};
\node [internal, scale=0.5] (2) at (6, 2) {};
\node at (5.5, 3) {\tiny{1}};
\node at (5.5, 1) {\tiny{2}};
\node at (6.5, 3) {\tiny{4}};
\node at (6.5, 1) {\tiny{3}};
\node [external, scale=0.5] (3) at (9, 2) {};
\node [internal, scale=0.5] (4) at (12, 2) {};
\node at (11.5, 3) {\tiny{1}};
\node at (11.5, 1) {\tiny{2}};
\node at (12.5, 3) {\tiny{4}};
\node at (12.5, 1) {\tiny{3}};
\node [external] (4a) at (13.5, 3) {};
\node [external] (4b) at (13.5, 1) {};
\draw (0) to [bend left=40] (1) to [bend right=45] (2)[postaction={decorate, decoration={
                                        markings,
                                        mark=at position 0.2 with {\arrow[>=square,white, scale=0.7] {>}; },
                                        mark=at position 0.2 with {\arrow[>=open square, scale=0.7]  {>}; } } }];
\draw (0) to [bend right=40] (1) to [bend left=45] (2)[postaction={decorate, decoration={
                                        markings,
                                        mark=at position 0.2 with {\arrow[>=square,white, scale=0.7] {>}; },
                                        mark=at position 0.2 with {\arrow[>=open square, scale=0.7]  {>}; } } }];
\draw (2) to [bend left=40] (3) to [bend right=45] (4)[postaction={decorate, decoration={
                                        markings,
                                        mark=at position 0.2 with {\arrow[>=square,white, scale=0.7] {>}; },
                                        mark=at position 0.2 with {\arrow[>=open square, scale=0.7]  {>}; } } }];
\draw (2) to [bend right=40] (3) to [bend left=45] (4)[postaction={decorate, decoration={
                                        markings,
                                        mark=at position 0.2 with {\arrow[>=square,white, scale=0.7] {>}; },
                                        mark=at position 0.2 with {\arrow[>=open square, scale=0.7]  {>}; } } }];
\draw (0) to [bend right=15] (0a);
\draw (0) to [bend left=15] (0b);
\draw (4) to [bend right=15] (4b);
\draw (4) to [bend left=15] (4a);
\end{tikzpicture}
}
\subfloat[${\scriptstyle M(NM')^{s-1}}$]{
\begin{tikzpicture}[scale=0.3]
\node [internal, scale=0.5] (0) at (0, 2) {};
\node at (-0.5, 3) {\tiny{1}};
\node at (-0.5, 1) {\tiny{2}};
\node at (0.5, 3) {\tiny{4}};
\node at (0.5, 1) {\tiny{3}};
\node [external] (1) at (3, 2) {};
\node [external] (0a) at (-1.5, 3) {};
\node [external] (0b) at (-1.3, 1) {};
\node [external, scale=0.5] (2) at (6, 2) {};
\node [external, scale=0.5] (2x) at (5.5, 2.3) {};
\node [external, scale=0.5] (2y) at (5.5, 1.7) {};
\node [external, scale=0.5] (2z) at (6.5, 2.3) {};
\node [external, scale=0.5] (2w) at (6.5, 1.7) {};
\node  at (6, 2) {$\ldots$};
\node [external, scale=0.5] (3) at (9, 2) {};
\node [internal, scale=0.5] (4) at (12, 2) {};
\node at (11.5, 3) {\tiny{1}};
\node at (11.5, 1) {\tiny{2}};
\node at (12.5, 3) {\tiny{4}};
\node at (12.5, 1) {\tiny{3}};
\node [external] (4a) at (13.5, 3) {};
\node [external] (4b) at (13.5, 1) {};
\draw (0) to [bend left=40] (1) to [bend right=45] (2y)[postaction={decorate, decoration={
                                        markings,
                                        mark=at position 0.2 with {\arrow[>=square,white, scale=0.7] {>}; },
                                        mark=at position 0.2 with {\arrow[>=open square, scale=0.7]  {>}; } } }];
\draw (0) to [bend right=40] (1) to [bend left=45] (2x)[postaction={decorate, decoration={
                                        markings,
                                        mark=at position 0.2 with {\arrow[>=square,white, scale=0.7] {>}; },
                                        mark=at position 0.2 with {\arrow[>=open square, scale=0.7]  {>}; } } }];
\draw (2z) to [bend left=40] (3) to [bend right=45] (4)[postaction={decorate, decoration={
                                        markings,
                                        mark=at position 0.2 with {\arrow[>=square,white, scale=0.7] {>}; },
                                        mark=at position 0.2 with {\arrow[>=open square, scale=0.7]  {>}; } } }];
\draw (2w) to [bend right=40] (3) to [bend left=45] (4)[postaction={decorate, decoration={
                                        markings,
                                        mark=at position 0.2 with {\arrow[>=square,white, scale=0.7] {>}; },
                                        mark=at position 0.2 with {\arrow[>=open square, scale=0.7]  {>}; } } }];
\draw (0) to [bend right=15] (0a);
\draw (0) to [bend left=15] (0b);
\draw (4) to [bend right=15] (4b);
\draw (4) to [bend left=15] (4a);
\end{tikzpicture}
}
\caption{\scriptsize{Recursive construction of the interpolation in Lemma~\ref{3twin}.
The circles are assigned $f$ and the squares are assigned $\neq_2$.}}
\label{interpolation}
\end{figure}

\clearpage

\begin{proof}
We treat
$\sum_{j,k \ge 0,~j+k\leq m} (\alpha^j \beta^k)^{\ell} x_{j,k} = N_{\ell}$
(where $1 \le {\ell} \le {m+2 \choose 2}$)
as a system of linear equations with unknowns $x_{j,k}$.
The coefficient vector of the first equation is $(\alpha^j \beta^k)$,
indexed by the pair $(j,k)$, where $0 \le j, k \le m$ and $j+k\leq m$.
 The coefficient matrix of the linear system is a Vandermonde matrix,
with row index ${\ell}$ and column index $(j,k)$.
However, this Vandermonde matrix is rank deficient.
If $(j,k)-(j',k') \in L$, then columns $(j,k)$ and $(j',k')$ have the same value.

We can combine the identical columns $(j,k)$ and $(j',k')$
if $(j,k)-(j',k') \in L$, since for each coset $T$ of $L$,
the value $\alpha^j \beta^k$ is constant.  Thus, the sum
$\sum_{j,k \ge 0,~j+k\leq m} (\alpha^j \beta^k)^{\ell} x_{j,k}$ can be
written as
$\sum_{T} (\alpha^j \beta^k)^{\ell}
 \left(
\sum_{j,k \ge 0,~j+k\leq m,~
(j,k) \in T}   x_{j,k} \right)$,
where the sum over $T$ is for all cosets $T$ of $L$
having a non-empty intersection
with the cone $C = \{ (j,k) \mid 0 \le j, k \le m, ~j+k\leq m\}$.
Now  the coefficient matrix, indexed by $1 \le {\ell} \le {m+2 \choose 2}$
for the  rows  and the cosets $T$ with $T \cap C \not = \emptyset$
for the columns,
has full rank.
And so we can solve $\left( \sum_{j,k \ge 0,~j+k\leq m,~
(j,k) \in T}   x_{j,k} \right)$ for each coset $T$ with
$T \cap C \not = \emptyset$.
 Notice that for the sum
 $\sum_{j+k\leq m} \phi^j \psi^k x_{j,k}$, we also have the
expression
 $\sum_{T}  \phi^j \psi^k
 \left(
\sum_{j,k \ge 0,~j+k\leq m,~
(j,k) \in T}   x_{j,k} \right)$, since $\phi^j \psi^k$
on each coset $T$ of $L$ is also constant.
The lemma follows.
\end{proof}

Now we prove the \#P-hardness for the 3-twins case.
In this case $a=x$, $b=y$ and $c=z$. We denote by $M(a,b,c)$
the problem defined by the
 signature matrix $M_{x_1x_2, x_4x_3} (a,a,b,b,c,c)$.

\begin{lemma}\label{3twin}
Let $f$ be a 4-ary signature with the signature matrix
$M_{x_1x_2, x_4x_3}(f)=
\left [ \begin{smallmatrix}
0 & 0 & 0 & a\\
0 & b & c & 0\\
0 & c & b & 0\\
a & 0 & 0 & 0\\
\end{smallmatrix} \right ]$ with $abc\neq 0$.
Then $\operatorname{Holant}(\neq_2|f)$ is \#P-hard.
\end{lemma}

\begin{proof}
We construct a series of
   gadgets by a chain of one leading copy of $f$ and
a sequence of twisted copies of $f$ linked by two $(\neq_2)$'s
in between.
It has the signature matrix $D_s= M (N M')^{s-1}$,
for $s \ge 1$,
where  $M =M_{x_1x_2,x_4x_3}(f)$,
$M' = M_{x_2x_1,x_4x_3}(f)$ is a permuted copy of $M$,
and $N$ is the double {\sc Disequality}.
See Figure~\ref{interpolation}.
This is in the right side of
 $\operatorname{Holant}(\neq_2|f)$.

The signature matrix  of this gadget is
 given as a product of matrices. Each matrix is a function of arity $4$.
Notice that the two row indices in $M_{x_2x_1,x_4x_3}(f)$ exchange their positions compared with the standard one $M_{x_1x_2,x_4x_3}(f)$.
%
%
%
Thus the rows of $M$ under go the permutation
$(00, 01, 10, 11) \rightarrow (00, 10, 01, 11)$ to get $M'$.
In other words, $M'$
is obtained from $M$ by exchanging the middle two rows.
Also $NM'$ reverses all 4 rows of $M'$. So we have
\[N M'=
\left [ \begin{array}{cccc}
a & 0 & 0 & 0 \\
0 & b & c & 0 \\
0 & c & b & 0 \\
0 & 0 & 0 & a
\end{array} \right ], ~~~~\mbox{and}~~~~ D_s=
\left [ \begin{array}{cccc}
0 & \bf{0} & a^s \\
\mathbf{0} & {\left [ \begin{array}{cc}  b & c \\ c & b \end{array} \right ]^s} & \mathbf{0} \\
a^s & \bf{0} & 0
\end{array} \right ].\]

We diagonalize the 2 by 2 matrix in the middle using
$H=\frac{1}{\sqrt{2}}
\left[\begin{smallmatrix}
1 & 1 \\
 1 & -1
\end{smallmatrix}\right]
$ (note that $H^{-1} = H$),
 and get $D_s=P \Lambda_s P$, where
\[P=\left [ \begin{array}{ccc}
1 & \bf{0} & 0 \\
\mathbf{0} & H & \mathbf{0} \\
0 & \bf{0} & 1
\end{array} \right ], ~~~~\mbox{and}~~~~
\Lambda_s=
\left [ \begin{array}{cccc}
0 & 0 & 0 & a^s \\
0 & (b+c)^s & 0 & 0 \\
0 & 0 & (b-c)^s & 0 \\
a^s & 0 & 0 & 0
\end{array} \right ].\]

The matrix $\Lambda_s$ has a good form for polynomial interpolation. Suppose we have a problem $\operatorname{Holant}(\neq_2\mid F)$ to be reduced to $\operatorname{Holant}(\neq_2\mid M)$. Let $F$ appear $m$ times in an instance $\Omega$. We replace each appearance of $F$ by a copy of
the gadget $D_s$, to get an instance $\Omega_s$ of $\operatorname{Holant}(\neq_2\mid M)$.
We can treat each of the $m$ appearances of $D_s$ as a new gadget composed of three functions in sequence
$P$, $\Lambda_s$ and $P$, and denote this new instance by $\Omega'_s$. We divide  $\Omega'_s$ into two parts. One part is composed of $m$ functions $\Lambda_s$. The second part is the rest of the functions, including $2m$ occurrences
of $P$, and its signature is represented by $X$ (which is a tensor expressed
as a row vector).
 The Holant value of  $\Omega'_s$ is the dot product $\langle X, \Lambda_s^{\otimes m}  \rangle$, which is a summation over $4m$ bits, that is, the values of the $4m$ edges connecting the two parts. We can  stratify
all 0-1 assignments of these  $4m$ bits having a nonzero
evaluation of $\operatorname{Holant}_{\Omega'_s}$ into the following categories:
\begin{itemize}
\item
There are $i$ many copies of $\Lambda_s$ receiving inputs $0011$ or $1100$;
\item
There are $j$ many copies of $\Lambda_s$ receiving inputs $0110$; and
\item
There are $k$ many copies of $\Lambda_s$ receiving inputs $1001$
\end{itemize}
such that $i+j+k = m$.

For any assignment in the category with parameter $(i,j,k)$, the evaluation of
$\Lambda_s^{\otimes m}$ is clearly $a^{si} (b+c)^{sj}  (b-c)^{sk}$.
 We can rewrite the dot product summation and  get
\begin{equation}  \label{eqa: interpo}
 \operatorname{Holant}_{\Omega_s} =
 \operatorname{Holant}_{\Omega'_s} =
\langle X, \Lambda_s^{\otimes m}  \rangle
=\sum_{i+j+k=m}  a^{si} (b+c)^{sj}  (b-c)^{sk}  x_{i,j,k},
\end{equation}
where $x_{i,j,k}$ is the summation of values of
the second part $X$ over all assignments in the category $(i,j,k)$.
Because $i+j+k=m$, we also use $x_{i,j}$
to denote the value $x_{i,j,k}$.
Similarly we use  $x_{j,k}$ or $x_{i,k}$ to denote the same value $x_{i,j,k}$
when there is no confusion.

Generally, in an interpolation reduction, we pick polynomially many values of $s$, and get a system of linear equations in $x_{i,j,k}$.
When all $a^i (b+c)^j (b-c)^k$are distinct, for $i+j+k=m$, we get a full rank Vandermonde coefficient matrix, and then we can solve for
each $x_{i,j,k}$. Once we have $x_{i,j,k}$ we can compute any
function in $x_{i,j,k}$.

When $a^i (b+c)^j (b-c)^k$ are not distinct, say $a^{i} (b+c)^{j} (b-c)^{k}
=a^{i'} (b+c)^{j'} (b-c)^{k'}$, we may define a new variable $y=x_{i,j,k}
+x_{i',j',k'}$. We can combine all $x_{i,j,k}$
with the same $a^i (b+c)^j (b-c)^k$.
Then we have a
full rank Vandermonde system of linear equations in these new unknowns.
We can solve all new unknowns
 and then sum them up to get $\sum_{i+j+k=m}  x_{i,j,k}$.
This is one special  function in $x_{i,j,k}$.

The above are two typical application methods in this kind of interpolation.
Unfortunately in our case, we may have a rank deficient Vandermonde system,
and the sum  $\sum_{i+j+k=m}  x_{i,j,k}$ does not give us anything
useful.
This is because if we
replace $a^{si} (b+c)^{sj}  (b-c)^{sk}$ by the constant value $1$
in equation (\ref{eqa: interpo}), we get $\sum_{i+j+k=m}  x_{i,j,k}$.
Thus,
$\sum_{i+j+k=m}  x_{i,j,k}$ corresponds to $\Omega'_s$ with
 all nonzero values in $\Lambda_s$ replaced by the constant $1$,
i.e., we get a reduction from the problem $M(1,1,0)$.
%
But  $M(1,1,0)$ is a tractable problem, and so
we do not get any hardness result by such a reduction.


To prove this lemma, there are three cases when there are 3 twins.
\begin{enumerate}
\item
Two elements in $\{a,b,c\}$ are equal.
By the symmetry of the group action of $S_4$,
without loss of generality, we may assume $b=c$. We have

\[\Lambda_s=
\left [ \begin{array}{cccc}
0 & 0 & 0 & a^s \\
0 & (2b)^s & 0 & 0 \\
0 & 0 & 0 & 0 \\
a^s & 0 & 0 & 0
\end{array} \right ],
\] and equation (\ref{eqa: interpo}) becomes
$\operatorname{Holant}_{\Omega_s}
=\sum_{i+j=m}  x_{i,j} a^{si} (2b)^{sj}$.
Note that all terms $x_{i,j,k}$ with $k\not = 0$ have disappeared.
We can interpolate to get
$\sum_{i+j=m} x_{i,j}$. This sum corresponds to a \#P-hard
problem. In  fact we define $A=
\left [ \begin{smallmatrix}
0 & 0 & 0 & 1 \\
0 & 1 & 0 & 0 \\
0 & 0 & 0 & 0 \\
1 & 0 & 0 & 0
\end{smallmatrix} \right ]$, and then
$PAP= \left [ \begin{smallmatrix}
0 & 0 & 0 & 1 \\
0 & \frac{1}{2} & \frac{1}{2} & 0 \\
0 & \frac{1}{2} & \frac{1}{2} & 0 \\
1 & 0 & 0 & 0
\end{smallmatrix} \right ]$.
Then $\operatorname{Holant}(\neq_2 \mid M(1,\frac{1}{2},\frac{1}{2}))$
 is \#P-hard by the determinant criterion for redundant matrices,
Theorem~\ref{redundant}.

\item
Two elements in $\{a,b,c\}$ have the opposite value.
By the symmetry of the group action of $S_4$,
without loss of generality, we may assume $b=-c$. We have
\[\Lambda_s=
\left [ \begin{array}{cccc}
0 & 0 & 0 & a^s \\
0 & 0 & 0 & 0 \\
0 & 0 & (2b)^s & 0 \\
a^s & 0 & 0 & 0
\end{array} \right ],
\]
 and equation (\ref{eqa: interpo}) becomes
$\operatorname{Holant}_{\Omega_s}
=\sum_{i+k=m}  x_{i,k} a^{si} (2b)^{sk}$.
Simiarly note that all terms $x_{i,j,k}$ with $j\not = 0$ have disappeared.
  We can interpolate to get
$\sum_{i+k=m} x_{i,k}$. Simiarly
we show that this sum corresponds to a \#P-hard
problem. In  fact we define $B=
 \left [ \begin{smallmatrix}
0 & 0 & 0 & 1 \\
0 & 0 & 0 & 0 \\
0 & 0 & 1 & 0 \\
1 & 0 & 0 & 0
\end{smallmatrix} \right ]$.
Then $PBP=
\left [ \begin{smallmatrix}
0 & 0 & 0 & 1 \\
0 & \frac{1}{2} & -\frac{1}{2} & 0 \\
0 & -\frac{1}{2} & \frac{1}{2} & 0 \\
1 & 0 & 0 & 0
\end{smallmatrix} \right ]$.
This matrix defines
the problem  $\operatorname{Holant}(\neq_2 \mid M(2,1,-1))$,
up to a  nonzero constant factor.

By the group action we also have $M(-1,2,1)$.
If we link two copies of $M(-1,2,1)$ by $N$, we get
$M(1,5,4)$, because
$\left[\begin{smallmatrix}
2 & 1 \\
 1 & 2
\end{smallmatrix}\right]^2
=
\left[\begin{smallmatrix}
5 & 4 \\
4 & 5
\end{smallmatrix}\right]$.


Then $M(1,5,4)=P \Lambda P$, where
$\Lambda=
\left [ \begin{smallmatrix}
0 & 0 & 0 & 1 \\
0 & 9 & 0 & 0 \\
0 & 0 & 1 & 0 \\
1 & 0 & 0 & 0
\end{smallmatrix} \right ]$.

There are only two nonzero values $9$ and $1$ in $\Lambda$.
For $M(1,5,4)$, we have
$\operatorname{Holant}_{\Omega_s}
=\sum_{0 \le i \le m}  x_{i} 9^{si}$, from which we can solve all $x_i$
($i=0,1,\ldots,m$), we can compute $\sum_{0 \le i \le m}  x_{i} 3^{si}$.
This realizes the following problem $
\left [ \begin{smallmatrix}
0 & 0 & 0 & 1 \\
0 & 3 & 0 & 0 \\
0 & 0 & 1 & 0 \\
1 & 0 & 0 & 0
\end{smallmatrix} \right ]$, which gives us $\operatorname{Holant}(\neq_2
\mid M(1,2,1) )$.
By the symmetry of group action
we also have $\operatorname{Holant}(\neq_2 \mid
M(2,1,1))$, which is  \#P-hard by
Theorem~\ref{redundant}.

\item
If we consider $a$, $b$ and  $c$ as three nonzero complex numbers
on the plane, there
 are two elements in $\{a,b,c\}$ which are not orthogonal as vectors.
By the symmetry of group action of $S_4$,
we may assume $b$ and $c$ are not orthogonal.
If $b+c =0$ or $b-c = 0$, then it is already proved in the first two cases.
So we may assume $b \not = \pm c$.

By the interpolation method, we have a system of linear
 equations in $x_{i,j,k}$,
whose coefficient matrix  $( (a^i (b+c)^j (b-c)^k)^s )$ has
 row index $s$ and column index from $\{(i,j,k) \mid i,j,k \in \mathbb{N},  i+j+k=m\}$.

Let $\alpha=\frac{b+c}{a}$ and $\beta=\frac{b-c}{a}$.
Then  they have different norms $|\alpha| \neq |\beta|$.
Indeed, if $|\alpha| = |\beta|$ then $|1 +c/b| = |1 - c/b|$
which means that $c/b \in {\frak i} \mathbb{R}$ is purely imaginary,
i.e., $b$ and $c$ are orthogonal.

The matrix $( (a^i (b+c)^j (b-c)^k)^s )$,
after dividing the $s$th row by $a^{sm}$,
has the form $( (\alpha^{j} \beta^k)^s)$, which
 is a Vandermonde matrix with  row index $s$ and column index
from $\{(j,k) \mid j,k \in \mathbb{N},  j+k \leq m\}$.
Define $L=\{(j,k) \in \mathbb{Z}^2 \mid \alpha^j \beta^k=1 \}$.
This is a sublattice of $\mathbb{Z}^2$. Every lattice has a basis.
There are 3 cases depending on the rank of $L$.
\begin{enumerate}
\item
$L=\{(0,0)\}$. All $\alpha^j \beta^k$ are distinct. It is an interpolation reduction in full power. We can realize $
\left [ \begin{smallmatrix}
0 & 0 & 0 & 1 \\
0 & 3 & 0 & 0 \\
0 & 0 & 1 & 0 \\
1 & 0 & 0 & 0
\end{smallmatrix} \right ]$. This  corresponds to
 $\operatorname{Holant}(\neq_2 \mid
M(2,1,1))$, which is \#P-hard by Theorem~\ref{redundant}.

\item
$L$ contains two  vectors $(j_1,k_1)$ and $(j_2,k_2)$
independent over $\mathbb{Q}$.
Then the nonzero vectors  $j_2 (j_1,k_1)-j_1 (j_2,k_2)=(0,j_2k_1-j_1k_2)$ and
$k_2 (j_1,k_1)-k_1 (j_2,k_2)=(k_2j_1-k_1 j_2,0)$ are in $L$. Hence, both $\alpha$ and $\beta$ are roots of unity, but this contradicts with $|{\alpha}| \neq
|{\beta}|$.

\item
$L=\{(ns, nt) \mid n \in \mathbb{Z}\}$, where $s,t \in \mathbb{Z}$ and
 $(s,t) \neq (0,0)$.
We know that $s+t \neq 0$, otherwise we get $|{\alpha}| \neq
|{\beta}|$.
By Lemma~\ref{lem:interpolation}, for any numbers $\phi$ and $\psi$ satisfying $\phi^s \psi^t=1$, we can compute $\sum_{j+k\leq m} \phi^j \psi^k x_{j,k}$ efficiently.

Define $A=
\left [ \begin{smallmatrix}
0 & 0 & 0 & 1 \\
0 & \phi & 0 & 0 \\
0 & 0 & \psi & 0 \\
1 & 0 & 0 & 0
\end{smallmatrix} \right ]$,  and we have
$2PAP= \left [ \begin{smallmatrix}
0 & 0 & 0 & 2 \\
0 & \phi+\psi & \phi-\psi & 0 \\
0 & \phi-\psi & \phi+\psi & 0 \\
2 & 0 & 0 & 0
\end{smallmatrix} \right ]$.
We get
$\operatorname{Holant}(\neq_2 \mid
M(2,\phi+\psi,\phi-\psi))$.

\begin{enumerate}
\item
$t=0$. Without loss of generality  $s>0$. Let $\phi=1$ and $\psi=1/2$.
 We get $M(4,3,1)$, from which we can get  $M(1,4,3)$
by the $S_4$ group symmetry.
This is \#P-hard by the same proof method as we prove $M(1,5,4)$ is
\#P-hard in Case 2.

\item
$t>0$ and $s \geq 0$. Let $\phi=\psi+2$. We need $f(\psi)=(\psi+2)^s \psi^t=1$. Because $f(0)=0<1$ and $f(1)\geq 1$, there is a root $\psi_0 \in (0,1]$. We get $M(2,2\psi_0+2,2)$, which is \#P-hard by Case 1.

\item
$t>0$, $s<0$ and $|t|>|s|$. Let $\phi=\psi+2$. $\psi^{|t|}=(\psi+2)^{|s|}$ has a solution $\psi_0$ in $(1,\infty)$. We get $M(2,2\psi_0+2,2)$, which is \#P-hard by Case 1.

\item
$t>0$, $s<0$ and $|t|<|s|$. Let $\psi=\phi+2$. $\phi^{|s|}=(\phi+2)^{|t|}$ has a solution $\phi_0$ in $(1,\infty)$. We get $M(2,2\phi_0+2,-2)$, which is \#P-hard by Case 2.
\end{enumerate}
\end{enumerate}
\end{enumerate}
\end{proof}

We finish this section by proving the other no zero cases can realize 3-twins.

\begin{lemma}  \label{no zero}
Let $f$ be a 4-ary signature with the signature matrix
$M_{x_1x_2, x_4x_3}(f)=\left[\begin{smallmatrix}
0 & 0 & 0 & a\\
0 & b & c & 0\\
0 & z & y & 0\\
x & 0 & 0 & 0\\
\end{smallmatrix}\right]$ with $abcxyz\neq 0$.
Then $\operatorname{Holant}(\neq_2\mid f)$ is \#P-hard.
\end{lemma}

\begin{proof}
Note that
$M_{x_4x_3, x_1x_2}(f)=
\left[\begin{smallmatrix}
0 & 0 & 0 & x \\
0 & b & z & 0 \\
0 & c & y & 0 \\
a & 0 & 0 & 0
\end{smallmatrix}\right]$.
Connecting two copies of $f$ back to back
 by double {\sc Disequality} $N$, we get the gadget whose signature has the signature matrix
\begin{center}
$M_{x_1x_2, x_4x_3}(f)NM_{x_4x_3, x_1x_2}(f)=\left [ \begin{array}{cccc}
0 & 0 & 0 & ax \\
0 & 2bc & by+cz & 0 \\
0 & by+cz & 2yz & 0 \\
ax & 0 & 0 & 0
\end{array} \right ].$
\end{center}

If $by+cz \neq 0$, we have realized a function $M(ax,ax,2bc,2yz,by+cz,by+cz)$ of two twins, with all nonzero values.
We can use $M(2bc,2yz,ax,ax,by+cz,by+cz)$ to construct the following function by the same gadget \[M(4bcyz,4bcyz,2ax(by+cz),2ax(by+cz),a^2x^2+(by+cz)^2,a^2x^2+(by+cz)^2).\]
If furthermore $a^2x^2+(by+cz)^2 \neq 0$, we get a nonzero 3-twins function and we can finish the proof by Lemma~\ref{3twin}.
If this process fails, we get a condition that either $by+cz=0$  or $\mathfrak{i}ax+by+cz=0$ or  $-\mathfrak{i}ax+by+cz=0$.
Recall the symmetry among the 3 pairs $(a,x), (b,y), (c,z)$. If we apply this process with a permuted form of $M$, we will get either $ax+cz=0$  or $ax+\mathfrak{i}by+cz=0$ or  $ax-\mathfrak{i}by+cz=0$.
There is one more permutation of $M$ which
gives us either $ax+by=0$  or $ax+by+\mathfrak{i}cz=0$ or  $ax+by-\mathfrak{i}cz=0$.

We claim that, when $axbycz \neq 0$, the 3 Boolean disjunction
 conditions can not hold simultaneously. Hence,
one of three constructions will succeed and give us \#P-hardness.

To prove the claim, we assume
that all 3 disjunction conditions hold. Then we
get 3 conjunctions, each  a disjunction of 3 linear equations.
Each equation is a homogeneous linear equation on $(ax, by, cz)$.
The 3 equations in the first conjunction all have the form
$\alpha \cdot ax+1 \cdot by + 1 \cdot cz=0$ where
$\alpha \in \{0,\mathfrak{i},-\mathfrak{i}\}$.
Similarly the 3 equations in the second and third conjunction all have the form
$1 \cdot ax+ \beta \cdot by + 1 \cdot cz=0$
and
$1 \cdot ax+  1 \cdot by + \gamma \cdot cz=0$
respectively.
  If at least one equation holds in each of the  3 sets of linear equations with nonzero solution $(ax, by, cz)$, the following determinant
\begin{equation}\label{equ:det} \det \left [ \begin{array}{ccc}
\alpha & 1 & 1  \\
1 & \beta &  1 \\
1 & 1 & \gamma
\end{array} \right ]=0,  \end{equation}
for some $\alpha, \beta, \gamma \in \{0,\mathfrak{i},-\mathfrak{i}\}$.
However, there are no choices of $\alpha, \beta, \gamma \in \{0,\mathfrak{i},-\mathfrak{i}\}$ such that Equation (\ref{equ:det}) holds:
The determinant is $\alpha \beta \gamma - 2 - \alpha - \beta - \gamma$.
For $\alpha, \beta, \gamma \in \{0,\mathfrak{i},-\mathfrak{i}\}$, the norm
$|2 + \alpha +  \beta +  \gamma| \ge 2$,
but $|\alpha \beta \gamma| = 0$ or $1$.

\end{proof}

\section{Case~\ref{exactly-one-zero}: Exactly one zero}

\begin{lemma}\label{exact:one:zero}
Let $f$ be a 4-ary signature with the signature matrix
\begin{center}
$M_{x_1x_2, x_4x_3}(f)=\begin{bmatrix}
0 & 0 & 0 & a\\
0 & b & c & 0\\
0 & z & y & 0\\
x & 0 & 0 & 0\\
\end{bmatrix}$,
\end{center}
where there is exactly one of $\{a, b, c, x, y, z\}$ that is zero,
then $\operatorname{Holant}(\neq_2\mid f)$
is \#P-hard.
\end{lemma}

\begin{proof}
Without loss of generality, we can assume that $b=0$.
Note that
$M_{x_3x_4, x_1x_2}(f)=\left[\begin{smallmatrix}
0 & 0 & 0 & x\\
0 & c & y & 0\\
0 & 0 & z & 0\\
a & 0 & 0 & 0\\
\end{smallmatrix}\right]$.
Connecting a copy of $f$ with this via $N$, we get a
signature $g$ with signature matrix
\begin{center}
$M_{x_1x_2, x_4x_3}(f) N M_{x_3x_4, x_1x_2}(f)=\begin{bmatrix}
0 & 0 & 0 & ax\\
0 & c^2 & cy & 0\\
0 & cy & y^2+z^2 & 0\\
ax & 0 & 0 & 0\\
\end{bmatrix}$.
\end{center}
If $y^2+z^2 \neq 0$,
by Lemma~\ref{no zero}, Holant$(\neq_2\mid g)$ is \#P-hard.
Thus Holant$(\neq_2\mid f)$ is \#P-hard.
Otherwise, we have
\begin{center}
$y^2+z^2=0$.
\end{center}
Similarly, $M_{x_3x_4, x_1x_2}(f) N M_{x_4x_3, x_2x_1}(f)$ gives us
\begin{center}
$y^2+cz=0$.
\end{center}
$M_{x_4x_3, x_1x_2}(f) N M_{x_2x_1, x_4x_3}(f)$ gives us
\begin{center}
$y^2+c^2=0$.
\end{center}
From these equations, we get $c^2=z^2=cz=-y^2$.
This gives us $z=c$ and $y=\pm \mathfrak{i}c$, and
$M=M_{x_1x_2, x_4x_3}(f)=\left[\begin{smallmatrix}
0 & 0 & 0 & a\\
0 & 0 & c & 0\\
0 & c & \pm \mathfrak{i} c & 0\\
x & 0 & 0 & 0\\
\end{smallmatrix}\right].$

For this matrix $M$, we may construct $MNM^{\tt T}=\left[\begin{smallmatrix}
0 & 0 & 0 & ax\\
0 & 0 & c^2 & 0\\
0 & c^2 & \pm 2\mathfrak{i}c^2 & 0\\
ax & 0 & 0 & 0\\
\end{smallmatrix}\right]$.  Now we may repeat the construction from the beginning
 using $MNM^{\tt T}$ instead of $M$.
Because $(c^2)^2+(\pm 2\mathfrak{i}c^2)^2 \neq 0$, we get a function of 6 nonzero values. By Lemma~\ref{no zero}, $\operatorname{Holant}(\neq_2\mid f)$ is \#P-hard.

%
%
\end{proof}

\section{Case~\ref{exactly-two-zeros}: Exactly two zeros from distinct pairs}

\begin{lemma}\label{exact:two:zero}
Let $f$ be a 4-ary signature with the signature matrix
\begin{center}
$M_{x_1x_2, x_4x_3}(f)=\begin{bmatrix}
0 & 0 & 0 & a\\
0 & b & c & 0\\
0 & z & y & 0\\
x & 0 & 0 & 0\\
\end{bmatrix}$,
\end{center}
where there are exactly two zero entries in $\{a, b, c, x, y, z\}$
and they are from distinct pairs,
 then
 $\operatorname{Holant}(\neq_2\mid f)$ is \#P-hard.
\end{lemma}
\begin{proof}
Recall from Section~\ref{sec:preliminary} that
 we can
arbitrarily reorder the  three rows in
$\left[\begin{smallmatrix}
a & x \\
b & y \\
c & z
\end{smallmatrix}\right]$, and
we can also reverse arbitrary two rows.
Thus, we can assume that $ax\neq 0, bz\neq 0$ and $c=y=0$.
Note that $M_{x_1x_2, x_4x_3}(f)=\left[\begin{smallmatrix}
0 & 0 & 0 & a\\
0 & b & 0 & 0\\
0 & z & 0 & 0\\
x & 0 & 0 & 0\\
\end{smallmatrix}\right]$ and
$M_{x_3x_4, x_1x_2}(f)=\left[\begin{smallmatrix}
0 & 0 & 0 & x\\
0 & 0 & 0 & 0\\
0 & b & z & 0\\
a & 0 & 0 & 0\\
\end{smallmatrix}\right]$.
Take two copies of $f$. If
we connect the variables $x_4, x_3$ of the first function with the variables $x_3, x_4$ of the second function using $(\neq_2)$, we get a
signature $g$ with the signature matrix
\begin{center}
$M_{x_1x_2, x_4x_3}(f)NM_{x_3x_4, x_1x_2}(f)=
\begin{bmatrix}
0 & 0 & 0 & ax\\
0 & b^2 & bz & 0\\
0 & bz & z^2 & 0\\
ax & 0 & 0 & 0\\
\end{bmatrix}.
$
\end{center}
By Lemma~\ref{no zero}, Holant$(\neq_2\mid g)$ is \#P-hard.
Thus Holant$(\neq_2\mid f)$ is \#P-hard.
%
%
%
\end{proof}

\section{Case~\ref{one-zero-in-each-pair}: One zero in each pair}

\begin{lemma}
If there is one zero in each pair of $(a,x), (b,y), (c,z)$,
then $\operatorname{Holant}(\neq_2\mid f)$ is computable in
polynomial time.
\end{lemma}

\begin{proof}
We will list the three strings of weight 2 where $f$ may be nonzero, by the symmetry of the group action of $S_4$.
We may assume the first string is $\xi=0011$.
The second string $\eta$, being not complementary to $\xi$ and of weight two, we may assume it is $0101$.

The third string $\zeta$, being not complementary
of either $\xi$ or $\eta$, and of weight two,
must be either $0110$ or $1001$.
%
%
%
Hence, {\scriptsize $\begin{array}{ccccccc}
\xi & = & 0 & 0 & 1 & 1 \\
\eta & = & 0 & 1 & 0 & 1 \\
\zeta & = & 0 & 1&  1 & 0
\end{array} ~~ \text{or}~~ \begin{array}{ccccccc}
\xi & = & 0 & 0 & 1 & 1 \\
\eta & = & 0 & 1 & 0 & 1 \\
\zeta & = &  1 &0 & 0 & 1
\end{array}$.}

Then $f(x_1,x_2,x_3,x_4)=\mbox{{\sc Is-Zero}}(x_1) \cdot g(x_2,x_3,x_4)$ or
 $\mbox{{\sc Is-One}}(x_4) \cdot h(x_1,x_2,x_3)$,
where $h \in \mathscr{M}$ and $g \in \mathscr{M'}$.
Note that the {\sc Is-Zero} and {\sc Is-One} are both
unary functions and both belong to $\mathscr{M} \cap \mathscr{M'}$.
By Theorem \ref{holant*:dichotomy},
 $\operatorname{Holant}(\neq_2\mid f)$
 is computable in
polynomial time.
\end{proof}

\section{Acknowledgments}
 We sincerely thank Xi Chen and Pinyan Lu for their interest and comments.





\end{document}